\documentclass[a4paper,runningheads,USenglish]{llncs}

\newcommand{\LV}[1]{#1}
\newcommand{\SV}[1]{}

\usepackage[OT2,T1]{fontenc}
\usepackage[USenglish]{babel}

\usepackage{placeins}

\LV{\usepackage[appendix=inline]{apxproof}}
\SV{\usepackage[appendix=strip]{apxproof}}
\newtheoremrep{theorem}{Theorem}

\usepackage{lineno}
\LV{\nolinenumbers} %uncomment to disable line numbering
\SV{\nolinenumbers}

\usepackage{amsmath,amssymb,amsfonts}
\usepackage{graphicx}
\usepackage[colorinlistoftodos]{todonotes}
\usepackage{hyperref}
\usepackage{xspace}
\usepackage{comment}
\usepackage{cancel}
\usepackage{mathtools}

\def \Ra {\Rightarrow}
\def \emptyword{\lambda}

\newcommand{\subword}{\mathrm{sub}}

\title{On Computational Completeness of Semi-Conditional Matrix Grammars}

\author{Henning Fernau \inst{1}
\and\\ Lakshmanan Kuppusamy\inst{2}\thanks{Corresponding Author}%\orcidID{0000-0003-2358-905X} 
\and\\ 
Indhumathi Raman\inst{3}%\orcidID{0000-0002-0981-9165}
}
\authorrunning{H. Fernau, L. Kuppusamy, I. Raman}

\institute{Fachbereich 4 -- Abteilung Informatikwissenschaften, Universit\"at Trier, 54286 Trier, Germany, \email{fernau@uni-trier.de}
 \and School of Computer Science and Engineering, VIT University, Vellore - 632~014, India, \email{klakshma@vit.ac.in} \and 
 Department of Computing Technologies, SRM Institute of Science and Technology, \\ 
 Kattankulathur, Chennai - 603203, India. \email{indhumar2@srmist.edu.in}}  
 
\date{July 2024}

\begin{document}
\maketitle
\begin{abstract}
Matrix grammars are one of the first approaches ever proposed in regulated rewriting, prescribing that rules have to be applied in a certain order. \LV{Originally, they have been introduced by \'Abrah\'am on linguistic grounds. }In \LV{traditional }regulated rewriting, the most interesting case shows up when all rules are context-free. 
Typical descriptional complexity measures incorporate the number of nonterminals or the matrix length, i.e., the number of rules per matrix. When viewing matrices as program fragments, it becomes natural to consider additional applicability conditions for such matrices.
Here, we focus on attaching a permitting and a forbidden string to every matrix in a matrix grammar. The matrix is applicable to a sentential form~$w$ only if the permitting string is a subword in~$w$ and the forbidden string is not a subword in~$w$. We call such a grammar, where the application of a matrix is conditioned as described, a semi-conditional matrix grammar. We consider  $(1)$ the maximal lengths of permitting and forbidden strings, $(2)$ the number of nonterminals, $(3)$ the number of conditional matrices, $(4)$ the maximal length of any matrix and $(5)$ the number of conditional matrices with nonempty permitting and forbidden strings,  as the resources (descriptional complexity measures) of a semi-conditional matrix grammar.

\LV{In this paper, w}\SV{W}e show that certain semi-conditional matrix grammar families defined by restricting  resources can generate all \LV{of the }recursively enumerable languages.
\end{abstract}
%\keywords{Semi-Conditional \SV{\& Matrix }Grammars \LV{\and Matrix Grammars \and Computational Completeness} \and Descriptional Complexity}

\section{Introduction}

A matrix grammar (originally introduced by S.~\'Abrah\'am on linguistic grounds in~\cite{Abr65}) consists in sequences of context-free rules called matrices; when a matrix\LV{ is chosen to be} applied to a sentential form, all rules in the sequence are applied in the given order. It is often considered the most important variation of regulated rewriting\LV{ systems}. \LV{Although this type of control on the work of context-free grammars is one of the earliest extensions of context-free grammars presented in the literature~\cite{DasPau89}, matrix grammars still raise many interesting questions (some of these are discussed later in this section). }Matrix grammars with appearance checking, having \LV{(i) }three nonterminals, \LV{all of them used in appearing checking mode, (ii) four nonterminals and two of them used in appearance checking mode, }are computationally complete, i.e., they characterize $\mathrm{RE}$~\LV{\cite{Feretal07,Fer03b,FreMarPau2004}}\SV{\cite{Feretal07}}. However, the lengths of the matrices are unbounded. It is not clear how to restrict the matrix length while still bounding the number of nonterminals. It is\LV{ also} known that matrix grammars without appearance checking are not computationally complete; see~\cite{HauJan94}.

Descriptional complexity (within formal languages) focuses on the influence of syntactic parameters within automata or grammars. For matrix grammars, such parameters of interest could be (e.g.) the number of nonterminals or the number of matrices.
For formalisms characterizing~$\mathrm{RE}$, it has always been asked how small certain parameters could be while still maintaining computational completeness, see \cite{Fer2021} for a survey. We continue this line of research here.

%\paragraph{Semi-conditional grammars.}
In 1985, Gh. P\u aun introduced another variant of regulated context-free grammar called semi-conditional grammars~\cite{Pau85} where a permitting and a forbidden string, associated to each (context-free) rule, govern the applicability of said context-free rule. In~\cite{Pau85}, the author also introduced a combination of \LV{two regulation mechanisms, namely, }matrix and semi-conditional grammars, called semi-conditional matrix grammars\LV{, whose language classes were denoted as $\mathcal{KM}^\lambda(i,j)$}. T\LV{his means that t}o each matrix (containing sequences of possibly erasing context-free rules), a permitting string $w_+$ of length at most~$i$ and a forbidden string $w_-$ of length at most~$j$ are associated, and the said matrix is only applicable to the sentential form $w$ if $w_+$ is a substring of $w$ and if $w_-$ does not occur as a substring of~$w$. The ordered pair $(i,j)$ is called the \emph{degree} of the semi-conditional grammar.

\paragraph{Our contribution.} In this paper, we consider $(1)$ the degree $(i,j)$, $(2)$ the number~$n$ of nonterminals, $(3)$ the number~$m$ of conditional matrices and $(4)$ the maximal length~$\ell$ of any matrix as descriptional complexity measures of semi-conditional matrix grammars and we denote its corresponding language class by $\textrm{SCM}(i,j;n;m,\ell)$. Further, in every matrix, if either the permitting or the forbidden string is empty, then we call the semi-conditional grammar as \emph{simple} and its corresponding language class as $\textrm{SSCM}(i,j;n;m,\ell)$. In \cite{MedKop2004}, it has been proved that SSCM$(3,1;7;2,3)=\mathrm{RE}$. 
In this paper, we improve or complement this result by showing that each of  $\mathrm{SSCM}(2,1;5;3,2)$, $\mathrm{SSCM}(3,1;5;2,2)$, $\mathrm{SSCM}(3,1;4;3,3)$ is equivalent to the class $\mathrm{RE}$ of recursively enumerable languages.\LV{ In this case, we call these $\mathrm{SSCM}$ systems computationally complete.} We also show that certain semi-conditional matrix grammars with binary matrices (with matrix length at most~$2$) can generate $\mathrm{RE}$ if we relax on the simplicity condition of the matrices. Our results are surveyed in Table~\ref{table-our-results}.

\section{Definitions and Notation}

\LV{In the following, w}\SV{W}e assume the reader to be familiar with basic notions of formal languages. For a word $w\in V^*$, we call $y\in V^*$ a \emph{subword} (\LV{sometimes also called a }factor) of~$w$ if there are $x,z\in V^*$ such that $w=xyz$. Let $\subword(w)$ denote the set of all subwords of~$w$.

\LV{A semi-conditional matrix grammar  (denoted as $\mathrm{SCM}$) represents another combination of semi-conditional grammars and matrix grammars (see \cite{Pau85}).} 
\begin{definition}\label{def-SCM}
A semi-conditional matrix or $\mathrm{SCM}$-grammar~\cite{Pau85} is a quadruple $G = (V, T, M, S)$, where $V$ is an
alphabet, $T \subsetneq V$ contains the \emph{terminals}, $N\coloneqq V \setminus T$ contains the \emph{nonterminals}, 
and $S \in N$ is the \emph{start symbol}.\\ $M$ is a finite set of \emph{matrices with context conditions} which are of the form
$[(A_1 \to x_1), \dots, (A_\ell \to x_\ell),P,F]$,
where $A_i \in N$ and $x_i \in V^*$, $P,F \in V^+\cup \{\emptyset\}$.
\end{definition}
 Consider the matrix $r=[(A_1 \to x_1), \dots, (A_\ell \to x_\ell),P,F]$, the sets $P$ and $F$ are called the \emph{permitting} and \emph{forbidden} conditions,
respectively; $\ell$ is known as the \emph{length} of~$r$; $\emptyset\notin V$ is a special symbol, meaning that a condition is missing.\LV{\footnote{For $P$, the values  $P=\emptyword$ and $P=\emptyset$ are equivalently meaning that there is no permitting condition put on the current sentential form, while for $F$, the value $F=\emptyset$ means that no forbidden condition is put, while the value $F=\emptyword$ would mean that no continuation is possible at all.}} If  both $P=F=\emptyset$, then we term the matrix~$r$ as \emph{unconditional}, otherwise, we call it \emph{conditional}.
% The length~$i$ of the longest permitting condition and the length~$j$ of the longest forbidding condition represent the degree $(i,j)$ of $G$; 
If all matrices are unconditional, then the degree of~$G$ is $(0,0)$ using the standard property that $|\emptyset| =0$. For brevity, then we simplify $[(A \to x), \emptyset, \emptyset]$ to $A \to x$ hereafter. If at least one condition equals $\emptyset$, then the matrix is termed \emph{simple}. If all matrices are simple, then the SCM grammar itself is termed as a \emph{simple SCM} grammar and is denoted as \emph{SSCM grammar}.

Let $r=[(A_1 \to x_1, \dots , A_\ell \to x_\ell),P,F]\in M$, with $x,y \in V^*$. Then $x \Rightarrow_r y$ (or simply $x \Rightarrow_G y$ or  $x \Rightarrow y$) iff (1) $P\neq\emptyset$ implies $P \in \subword(x)$ and (2) $F\neq \emptyset$ implies $F \not\in \subword(x)$, and (3) the matrix $[(A_1 \to
x_1), \dots, (A_\ell \to x_\ell)]$ is applied to $x$ to get $y$. 
This means that there are sentential forms $x=y_0,y_1,\dots,y_\ell=y$ such that $y_i$ is obtained from $y_{i-1}$ by replacing one occurrence of $A_i$ in  $y_{i-1}$ by $x_i$, or, in other words, by applying the context-free rule $A_i\to x_i$ on $y_{i-1}$, for $i=1,\dots,\ell$.
Let $\Rightarrow_G^*$ denote the transitive and reflexive closure of $\Rightarrow_G$. The language of $G$, denoted as $L(G)$, is defined as $L(G) = \{y\in T^*\mid S \Rightarrow_G^* y\}$.

An $\mathrm{(S)SCM}$-grammar is said to be of \emph{degree} $(i,j)$, where $i, j \in \mathbb{N}$, if in every matrix rule $\{[(A_1 \to x_1), \dots, (A_\ell \to x_\ell)], \alpha, \beta\}$ of $M$ we have $|\alpha| \leq i$ and $|\beta| \leq j$. 
%In this context, $|\emptyset|=|\emptyword|=0$.
We denote by $\mathrm{(S)SCM}(i,j;n;m,\ell)$, a family of languages generated by $\mathrm{(S)SCM}$-grammars, where 
\begin{itemize}
    \item $(i,j)$ is upper-bounding their degree,
    \item $n$ is an upper bound on the number of nonterminals, 
    \item $m$ is an upper bound on the number of conditional matrices,
    \item $\ell$ is an upper bound on the number of rules in any matrix, i.e., upper-bounding the \emph{length} of any matrix.
\end{itemize}
\LV{If every matrix in a $\mathrm{(S)SCM}$-grammar has exactly one rule (that is, the maximum length~$\ell$ of any matrix is $1$), then the $\mathrm{(S)SCM}$ grammar corresponds to a (simple) semi-conditional ($\mathrm{(S)SC}$) grammar, which clearly means that $\mathrm{(S)SC}(i,j;n,m)$ equals $\mathrm{(S)SCM}(i,j;n,m,1)$. For comparison, some of the best results in this context are also mentioned in Table~\ref{table-our-results}.}

The only existing result in the domain of $\mathrm{SSCM}$ grammars is (see \cite[Theorem~3]{MedKop2004}) that $\mathrm{SSCM}(3,1;7;2,3)=\mathrm{RE}$. In fact, that result was stated a bit differently, because the authors called $\max(i,j)$ the degree of an $\mathrm{SSCM}$-grammar.

It seems to be very difficult to find computational completeness results for $\mathrm{SSCM}$-grammars with at most four nonterminals and matrix length at most two\LV{, keeping the other parameters small}.
As a restriction to binary matrices is (otherwise) a quite common normal form for matrix grammar\LV{ variation}s, we now allow ourselves to drop the simplicity condition. %\todo{add "and count such non-simple rules available in all matrices"} 
We are not aware of any other \LV{previous }descriptional complexity results \LV{concerning}\SV{for} $\mathrm{SCM}$  grammars.
In order to quantify of how much we violate the simplicity condition, \LV{in addition to the parameters described so far, }we account for the number of non-simple matrices $s$ as a sixth dimension in our notation. 
In this sense, we note that $\mathrm{SSCM}(i,j;n;m,\ell)=\mathrm{SCM}(i,j;n;m,\ell,0)$. In particular, we could refer to Theorem~\ref{thm:sscm-21532} in this new notation as $\mathrm{SCM}(2,1;5;3,2,0)=\mathrm{RE}$. This way of accounting for non-simple rules was introduced\LV{ into the area of descriptional complexity of rewriting grammars} in~\cite{FerKupRam2024a}. \LV{With this new notation, t}\SV{T}he results of this paper are tabulated in Table~\ref{table-our-results}. 
For comparison, we also include some results from the literature. \LV{These earlier results were written with different terminology. Especially, i}\SV{I}f no permitting context conditions are present, we arrive at a degree of $(0,j)$,\LV{ leading to particular}\SV{ i.e.,} generalized forbidding matrix grammars \cite{FerKupRam2024,Med90a}.

\begin{table}[tb]
    \begin{center}
\scalebox{0.95}{      \begin{tabular}{|p{1.0cm}|p{1.45cm}|p{2.05cm}|p{1.75cm}|p{1.99cm}|p{1.99cm}|}
\hline
%No.&
Degree $(i,j)$ &\small$\#$ Nonterminals $n$ &$\#$ Conditional Matrices $m$ &\small \mbox{Max.\,Matrix} Length $\ell$ &\small \mbox{$\#$~Non-}Simple rules $s$ &Reference\\
\hline
%1.&
(3,1)&7&2&3&0&\cite{MedKop2004}, Thm.~3\\
\hline\LV{\hline}\LV{\hline}
%2.&
(2,1)&5&3&2&0&Thm. \ref{thm:sscm-21532}\\ \hline
%3.&
(3,1)&5&2&2&0&Thm. \ref{thm:sscm-31522}\\ \hline
%4.&
(3,1)&4&3&3&0&Thm. \ref{thm:sscm-31433}\\ \hline
%5.&
(4,3)&4&7&2&6&Thm. \ref{thm:scm-434726}\\ \hline
%{6}.&(5,2)&4&4&2&3&T%hm \ref{thm:scm-524423}\\
%\hline
%6.&
(5,2)&4&7&2&4&Thm. \ref{thm:scm-524724}\\ \hline
%7.&
(6,3)&4&7&2&3&Thm. \ref{thm:scm-634723}\\ \hline
%8.&
(6,3)&3&*&2&4 &Thm. \ref{thm:scm-633-24}\\ \hline
%9.&
(7,2)&3&*&2&3 &Thm. \ref{thm:scm-723-23}\\ \hline \LV{\hline}\LV{\hline}
%10. & 
(0,2)&*&*&3&0 &  \cite{Pau85}, Thm.~4.4\\ \hline
%11. & 
(0,1)&*&*&2&0 & \cite{FerKupRam2024}, Thm.~2\\\hline
%12. & 
(0,*) & 3 & * & * & 0 & Thm. \ref{thm:sscm-0-3--}\\
\hline\LV{\hline}\LV{\hline}
\LV{(1,1) &*&*&1&0 & \cite{Mas2009a}, Cor.~1\\\hline}
\LV{(2,1) &9&8&1&0 & \cite{FerKOR2021a}, Thm. 3.1\\\hline}
\LV{(2,1) &7&6&1&3 & \cite{FerKupRam2024a}, Thm. 4\\\hline}
\LV{(3,1) & 7 & 7 & 1 & 0 & \cite{FerKOR2021a}, Thm. 3.2\\\hline}
\LV{(3,1) & 6 & 6 & 1 & 4 & \cite{FerKupRam2024a}, Thm. 3\\\hline}
\LV{(6,3) & 4 & 18 & 1 & 15 & \cite{FerKupRam2024a}, Thm. 1\\\hline}
\end{tabular}}
\caption{Results of this paper, plus predecessor results: $\mathrm{SCM}(i,j;n;m,l,s) = \mathrm{RE}$. The $*$ denotes that the respective parameter is unbounded.}
\label{table-our-results}
    \end{center}
\end{table}

\section{Normal Forms for Type-0 Grammars}

In \cite{FerKupRam2024,Gef91a,MasMed2007a}, different normal forms for type-0 grammars have been described. They all present grammars that contain very few nonterminals and are still able to generate every $\mathrm{RE}$ language.
\LV{They can be best explained by first introducing a normal form that is possibly the best known of Geffert's results from~\cite{Gef91a}.}

\begin{proposition}\label{propos:(5,2)-GNF}
For each $\mathrm{RE}$ language~$L$, $L\subseteq T^*$, there is a type-0 grammar of the form $G= (V, T, P\cup \{AB\to \emptyword, CD\to \emptyword\}, S)$ with $L(G)=L$, where~$P$ contains only context-free rules and $N\coloneqq V\setminus T=\{S, A, B, C,D\}$. \SV{$P$-rules are:}\LV{More specifically, the context-free rules are of the forms}
\begin{enumerate}
    \item \label{stage 1} $S\to uSa$ with $u\in\{A,C\}^+$ and $a\in T$,
    \item \label{stage 2} $S\to uSv$ or $S\to uv$  with $u\in\{A,C\}^+$ and $v\in \{B,D\}^*$.
\end{enumerate}
\end{proposition}

As this normal form uses five nonterminals and two non-context-free rules, we call it $(5,2)$-GNF for short.
Derivations are \LV{supposed to be }performed in three stages: \begin{itemize}
    \item In Stage~1, a terminal suffix is created that will be (finally) the terminal word that is produced. Here, only context-free rules from Item~\ref{stage 1} are used.
    \item In Stage~2, only context-free rules from Item~\ref{stage 2} are employed.
    \item Finally in Stage~3, the non-context-free deletion rules come into play. 
\end{itemize}
Sometimes Stages~1 and~2 are together addressed as Phase~1, so that Stage~3 is Phase~2.
Unfortunately, it is also possible to mix the first two stages. Geffert~\cite{Gef91a} could prove that such mixtures can never lead to terminal strings, but when we simulate normal form grammars in the following, we must also consider the corresponding sentential forms.
\LV{Nonetheless, the structure of sentential forms~$w$ derivable from~$S$ is very clear:}\SV{Any sentential form~$w$ derivable from~$S$ is from}\LV{ $w$ must belong to}
\begin{equation}\label{eq:SF-GNF-(5,2)}
 L_{(5,2)}\coloneqq \{A,C\}^*\{S,\emptyword\}(\{B,D\}^*\cup T)^*\,.
\end{equation}

Based on $(5,2)$-GNF, other normal forms can be derived that basically differ in choosing different encodings of the nonterminal parts.
For instance, by applying the morphism defined by $A\mapsto CAA$, $B\mapsto BBC$, $C\mapsto CA$, $D\mapsto BC$, $x\mapsto x$ for $x\in T\cup\{S\}$ to the right-hand side of the context-free rules, one arrives at:

\begin{proposition}\label{propos:(4,2)-GNF}
For each $\mathrm{RE}$ language~$L$, $L\subseteq T^*$, there is a type-0 grammar of the form $G= (V, T, P\cup \{AB\to \emptyword, CC\to \emptyword\}, S)$ with $L(G)=L$, where~$P$ contains only context-free rules and $N\coloneqq V\setminus T=\{S, A, B, C\}$. \SV{$P$-rules look like}\LV{More specifically, the context-free rules are of the forms}
\begin{enumerate}
    \item  $S\to uSa$ with $u\in\{CA,CAA\}^+$ and $a\in T$,
    \item  $S\to uSv$ or $S\to uv$  with $u\in\{CA,CAA\}^+$ and $v\in \{BC,BBC\}^*$.
\end{enumerate}
\end{proposition}
To differentiate from the previously stated normal form, we refer to this one as  $(4,2)$-GNF for short. Sentential forms~$w$ of this type of grammar belong to:
\begin{equation}\label{eq:SF-GNF-(4,2)}
 L_{(4,2)}\coloneqq \{CA,CAA\}^*\{S,\emptyword,CC\}(\{BC,BBC\}^*\cup T)^*\,.
\end{equation}
Observe that as long as $S$ is present (i.e., in Stages~1 or~2), the strings $CC$ or $AB$ cannot occur as subwords in any sentential form that is derivable from~$S$. Moreover, $BA$ is never a subsequence of any sentential form derivable from~$S$.

Another popular normal form can be called $(3,2)$-GNF. It is derived from $(5,2)$-GNF by the morphism $A\mapsto ABB$, $B\mapsto BA$, $C\mapsto AB$, $D\mapsto BBA$, $x\mapsto x$ for $x\in T\cup\{S\}$ to the right-hand side of the context-free rules. This yields:
\begin{proposition}\label{propos:(3,2)-GNF}
For each $\mathrm{RE}$ language~$L$, $L\subseteq T^*$, there is a type-0 grammar of the form $G= (V, T, P\cup \{AA\to \emptyword, BBB\to \emptyword\}, S)$ with $L(G)=L$, where~$P$ contains only context-free rules and $N\coloneqq V\setminus T=\{S, A, B\}$. \SV{$P$-rules look like}\LV{More specifically, the context-free rules are of the forms}
\begin{enumerate}
    \item  $S\to uSa$ with $u\in\{ABB,AB\}^+$ and $a\in T$,
    \item  $S\to uSv$ or $S\to uv$  with $u\in\{ABB,AB\}^+$ and $v\in \{BA,BBA\}^*$.
\end{enumerate}
\end{proposition}
\noindent
Sentential forms~$w$ of this type of grammar belong to:
\begin{equation}\label{eq:SF-GNF-(3,2)}
 L_{(3,2)}\coloneqq \{ABB,AB\}^*\{S,\emptyword,AA,ABA\}(\{BA,BBA\}^*\cup T)^*\,.
\end{equation}

Masopust and Meduna had a slightly different encoding idea in \cite{MasMed2007a}: keeping a middle marker $\$$ in the string allows encoding with only two different symbols ``elsewhere''. Now, to the deletion rules $AB\to\emptyword$ and $CD\to\emptyword$, resp., there correspond \emph{shrinking rules} $0\$0\to \$$ and $1\$1\to \$$, resp., and (only) finally $\$\to\emptyword$. We call the resulting normal form Masopust-Meduna normal form\LV{, or MMNF for short.
 This yields:}\SV{ (MMNF).}
\begin{proposition}\label{propos:MMNF}
For each $\mathrm{RE}$ language~$L$, $L\subseteq T^*$, there is a type-0 grammar of the form $G= (V, T, P\cup \{0\$0\to \$, 1\$1\to \$,\$\to\emptyword\}, S)$ with $L(G)=L$, where~$P$ contains only context-free rules and $N\coloneqq V\setminus T=\{S, 0, 1,\$\}$. \SV{$P$-rules look like}\LV{More specifically, the context-free rules are of the forms}
\begin{enumerate}
    \item  $S\to uSa$ with $u\in\{0,1\}^+$ and $a\in T$,
    \item  $S\to uSv$ or $S\to u\$v$  with $u\in\{0,1\}^+$ and $v\in \{0,1\}^*$.
\end{enumerate}
\end{proposition}
\noindent
Sentential forms~$w$ of this type of grammar belong to:
\begin{equation}\label{eq:SF-MMNF}
 L_{\text{MM}}\coloneqq \{0,1\}^*\{S,\emptyword,\$\}(\{0,1\}^*\cup T)^*\,.
\end{equation}
\LV{In fact, o}\SV{O}ther encodings are possible for this strategy, for instance, by requiring $u\in\{10,100\}^+$ and $v\in\{01,001\}^*$, one gets the property that the penultimate rule to be applied is $1\$1\to\$$. Such encodings give additional structure, as actually used in this paper. To differentiate this from MMNF\LV{ as introduced above}, let us call it \emph{strong} MMNF \SV{(sMMNF)}\LV{, or sMMNF for short}. Sentential forms~$w$ of this type of grammar belong to:
\begin{equation}\label{eq:SF-sMMNF}
 L_{\text{sMM}}\coloneqq \{10,100\}^*\{S,\emptyword,\$\}(\{01,001\}^*\cup T)^*\,.
\end{equation}

However, we are now giving a simplified version of what we called \emph{modified MMNF} in~\cite{FerKupRam2024} (or MMMNF for short) as follows.
\begin{proposition}\label{propos:MMMNF}
For each $\mathrm{RE}$ language~$L$, $L\subseteq T^*$, there is a type-0 grammar of the form $G= (V, T, P\cup \{0\$1\to \$, 1\$0\to \$,\$\to\emptyword\}, S)$ with $L(G)=L$, where~$P$ contains only context-free rules and $N\coloneqq V\setminus T=\{S, 0, 1,\$\}$.
\LV{More specifically, the context-free}\SV{The} rules in~$P$ are \LV{of the forms }as in Propos.~\ref{propos:MMNF}. Sentential forms of this type of grammar \LV{belong to}\SV{are in}~$L_{\text{MM}}$.
\end{proposition}

\section{Main Results for SSCM Grammars}

Our first theorem is the first one that deals with SSCM grammars of degree~$(2,1)$.

\begin{theorem}\label{thm:sscm-21532}
$\mathrm{SSCM}(2,1;5;3,2)=\mathrm{RE}$.
\end{theorem}
\begin{proof}
Let $L \in \mathrm{RE}$ be generated by a grammar in $(4,2)$-GNF \LV{of the form }$G= (V, T, P\cup \{AB\to \emptyword, CC\to \emptyword\}, S)$ such that $P$ contains only context-free rules and $N=V\setminus T=\{S, A, B, C\}$ (see Propos. \ref{propos:(4,2)-GNF}). 
Next, we define the $\mathrm{SSCM}$-grammar $G'= (V', T, P'\cup P'', S)$, where $V' = V \cup \{\#\}$ (assuming $\# \not\in V$), $P'$ contains the (single-rule) unconditional matrices of the form $[(S\to \alpha), \emptyset,\emptyset]$ whenever $S\to \alpha \in P$ and $P''$ contains the three (multi-rule and simple) conditional matrices 
$$\begin{array}{rl}
r1=&[(A\to \#),(B \to\#),\emptyset,\# ]\\
r2=&[(C\to \#),(C\to \#), \emptyset,\# ]\\
r3=&[(\#\to\lambda),(\# \to \lambda), \#\#,\emptyset]
\end{array}$$
Clearly, $G'$ has degree $(2,1)$, 5 nonterminals, and 3 conditional binary matrices.

We now show that $L(G')=L(G)$. Trivially, $r1$ (or $r2$, resp.) and then $r3$ can simulate $AB\to\emptyword$ (or $CC\to\emptyword$, resp.), so that $L(G)\subseteq L(G')$ is clear. For the reverse inclusion, we first make some observations concerning any sentential form~$w$ that is derivable in $G'$; they can be shown by easy inspection or induction and will be used in the following without special mentioning.
\begin{enumerate}
\item  If~$w$ contains no~$\#$, then all matrices but $r3$ may apply.
    \item If~$w$ contains any~$\#$, then only unconditional matrices or $r3$ may apply.
    \item If~$w$ contains no~$S$, then no sentential form~$w'$ derivable from~$w$ in $G'$ will ever contain any~$S$.
    \item If~$w$ contains any~$S$, then there is no second occurrence of~$S$ in~$w$.
\end{enumerate}

To prove $L(G)\supseteq L(G')$, we will actually show the following claim by induction on the length of the derivation of~$w$:
If $S\Ra_{G'}^*w$, then either (1) $S\Ra_G^*w$ or (2) there exist two occurrences of $\#$ in $w$, i.e., $w=w_1\#w_2\#w_3$ and there exists a sentential form $w'$ such that $S\Ra_G^*w'$ and either $w'=w_1Aw_2Bw_3$ or $w'=w_1Cw_2Cw_3$.\LV{

}
The claim is trivially true if $w=S$ (derivation length 0).

Consider some sentential form $w$ such that $S\Ra_{G'}^n w$ for some $n>0$.
Then, there is some $v$ such that $S\Ra_{G'}^{n-1}v\Ra_{G'}w$. 
By induction, we know that either (1) $S\Ra_G^*v$, i.e., $v\in L_{(4,2)}$ as defined in Eq.~\eqref{eq:SF-GNF-(4,2)}, or (2) there exist two occurrences of $\#$ in $v$, i.e., $v=v_1\#v_2\#v_3$ and there exists a sentential form $v'$ such that $S\Ra_G^*v'$ and either (a) $v'=v_1Av_2Bv_3$ or (b) $v'=v_1Cv_2Cv_3$.

Consider Case~(1). We might apply $r1$ to get~$w$. Then, we get $w=w_1\#w_2\#w_3$.
As $v$ is also a sentential form of~$G$, $BA$ is not a subsequence of~$v$.
Therefore, $v=w_1Aw_2Bw_3$ as claimed.
Similarly, we might apply $r2$ to get~$w$, arriving again at $w=w_1\#w_2\#w_3$. Now, $v=w_1Cw_2Cw_3$ as claimed.
Finally, if $S$ occurs in~$v$, we can also apply an unconditional matrix. As this (trivially) corresponds to applying a context-free rule, $S\Ra_G^*w$.

Consider Case~(2). Hence, $v=v_1\#v_2\#v_3$. If $S$ does not occur in~$v$, then 
we must apply matrix $r3$. This is only possible if $v_2=\emptyword$, i.e., if $v_2\neq\emptyword$, the derivation is stuck. We mark this observation as $[*]$ which we would recall again later. Hence, $v'=v_1ABv_3$ or $v'=v_1CCv_3$ are derivable in~$G$ by induction hypothesis.
Now, $v\Ra_{r3}w$ yields $v'\Ra_{G}w$ or $v'\Ra_{G}w$ by applying the deletion rules $AB\to\emptyword$ or $CC\to\emptyword$, respectively. Hence, $S\Ra_G^*w$.

In the following, when considering Cases (2a) and (2b) separately, we assume that $S$ occurs in~$v$. Recall that the position of~$S$ is then unique.

Consider Case (2a) first, i.e., there is sentential form $v'$ such that $S\Ra_G^*v'$ and $v'=v_1Av_2Bv_3$. 
As $S$ occurs in~$v$, it will occur in~$v_2$ by the structure of rules of $(4,2)$-GNF. Hence, none of the conditional matrices is applicable on~$v$. Rather, we have to apply an unconditional matrix (with context-free rule $S\to \alpha$) to $v$ in order to get $w$. Hence, $v=v_1\#v_2'Sv_2''\#v_3$, as $v_2=v_2'Sv_2''$ for some $v'_2,v_2''$. Then, $w=v_1\#v_2'\alpha v_2''\#v_3$.
Now, $v'\Ra_G w'$ with  $w'=v_1 Av_2'\alpha v_2''Bv_3$ by applying the context-free rule $S\to \alpha$. This shows the claim also in this subcase of Case (2a).

Now, consider Case (2b). This means that there is sentential form $v'$ such that $S\Ra_G^*v'$ and $v'=v_1Cv_2Cv_3$. If $S$ is contained in $v_2$, then the argument is analogous to the one of the previous paragraph. If $S$ is contained in $v_1$, then (with $v_1=v_1'Sv_1''$) applying an unconditional matrix (with context-free rule $S\to \alpha$) to $v$ would result in $w=v_1'\alpha v_1''\#v_2\#v_3$. Now, if one would apply $S\to \alpha$ on $v'$, we get $w'$, with $w'=v_1'\alpha v_1''Cv_2Cv_3$. This shows the claim also in this subcase of Case (2b). As a side-remark: Clearly, one can continue this argument until $\alpha$ does not contain~$S$ anymore. But then, the derivation will be stuck according to the analysis of $[*]$, because $v_2\neq\emptyword$ as $CC$ cannot be a substring of any sentential form derivable in Stages~1 or~2 in~$G$.
The case when $S$ is contained in~$v_3$ is analogous to the case when $S$ is contained in~$v_1$.

By induction, the claim follows. As for terminal strings~$w$, only Case (1) can happen,  $L(G)\supseteq L(G')$ can be inferred.\qed
\end{proof}

The following is a trade-off result of the parameters degree and number of nonterminals compared to the previous theorem. Moreover, it returns to the consideration of degree $(3,1)$, a case previously considered \LV{by A. Meduna and T. Kopeček }in~\cite{MedKop2004}.

\begin{theorem}\label{thm:sscm-31522}
$\mathrm{SSCM}(3,1;5;2,2)=\mathrm{RE}$.
\end{theorem}

\begin{proof} Let $L\subseteq T^*$ be any $\mathrm{RE}$ language. 
By Propos.~\ref{propos:MMMNF}, there is a type-0 grammar of the form $G= (V, T, P\cup \{0\$1\to \$, 1\$0\to \$,\$\to\emptyword\}, S)$ with $L(G)=L$ such that $P$ contains only context-free rules  with left-hand side~$S$ and $N\coloneqq V\setminus T=\{S, 0, 1,\$\}$, i.e., $G$ is in MMMNF.
We need one more additional nonterminal~$\#$ in the SSCM grammar~$G'$ that we describe next. As before, we take over all context-free rules of $G$ as single-rule unconditional matrices (that we do not count into our numbers). Also, we have the matrix $r_\$=[(\$ \to\emptyword),\emptyset,\emptyset]$.  The two non-context-free rules are simulated by the following conditional matrices.
\begin{eqnarray*}
r1 &=& [(0\to\#),(1\to\#),\emptyset,\#]\\
r2 &=& [(\#\to\emptyword),(\#\to\emptyword),\#\$\#,\emptyset]
\end{eqnarray*}
It should be clear how this simulation works, so that $L(G)\subseteq L(G')$ is obvious.

For the converse inclusion $L(G)\supseteq L(G')$, we can actually show the following claim by induction on the length of the derivation of~$w$:
If $S\Ra_{G'}^*w$, then either (1) $S\Ra_G^*w$ or (2) there exist two occurrences of $\#$ in $w$, i.e., $w=w_1\#w_2\#w_3$ and there exists a sentential form $w'$ such that $S\Ra_G^*w'$ and either $w'=w_10w_21w_3$ or $w'=w_11w_20w_3$. \SV{Details can be found in the long version of this paper. \qed\end{proof}}

\begin{toappendix}
\SV{\subsection{The Induction Proof of Correctness of Theorem~\ref{thm:sscm-31522}}}
The claim is trivially true if $w=S$ (derivation length 0).

Consider some sentential form $w$ such that $S\Ra_{G'}^n w$ for some $n>0$.
Then, there is some $v$ such that $S\Ra_{G'}^{n-1}v\Ra_{G'}w$. 
By induction, we know that either (1) $S\Ra_G^*v$, i.e., $v\in L_{\text{MMNF}}$ as defined in Eq.~\eqref{eq:SF-MMNF}, or (2) there exist two occurrences of $\#$ in $v$, i.e., $v=v_1\#v_2\#v_3$ and there exists a sentential form $v'$ such that $S\Ra_G^*v'$ and either (a) $v'=v_10v_21v_3$ or (b) $v'=v_11v_20v_3$.

In Case (1), we might apply an unconditional matrix in the step $v\Ra_{G'}w$, but this clearly corresponds to applying a context-free rule in~$G$, so that $S\Ra_G^*w$ follows. Alternatively, we can apply the conditional matrix $r1$.
Then, clearly one occurrence of~0 and one occurrence of~1 in~$v$ will be replaced by~$\#$ in order to obtain~$w$. Hence, we arrive at a string~$w$ with  two occurrences of $\#$, and we also have a sentential form $w'$ of $G$ that can be obtained from $w$ by appropriately replacing one $\#$ by~0 and the other $\#$ by~1. Namely, choose $w'=v$.

In Case (2), again we might apply  an unconditional matrix  in the step $v\Ra_{G'}w$. If we apply the corresponding context-free rule of~$G$ analogously to $v'$, we arrive at $w'$, i.e., $v'\Ra_G w'$ such that $w'$ testifies that~$w$ satisfies Condition (2). If we apply a conditional matrix instead, it must be matrix~$r2$. This means (by the permitting context) that $v=v_1\#\$\#v_2$.
Also, we know that $v'=v_10\$1v_2$ or $v'=v_11\$0v_2$ by induction. Therefore, $v'\Ra_Gw$ holds by applying one of the two non-context-free rules of~$G$.

By induction, the claim follows. As for terminal strings~$w$, only Case (1) can happen,  $L(G)\supseteq L(G')$ can be inferred.
\end{toappendix}
\LV{\qed \end{proof}}

\begin{toappendix}
\SV{\subsection{A Remark on Trade-off Results}}    
\begin{remark}
One could try to get a trade-off result to the previous result by modifying the construction as follows.
\begin{eqnarray*}
r1 &=& [(0\to\$\$),(1\to\$),\emptyset,\$\$]\\
r2 &=& [(\$\to\emptyword),(\$\to\emptyword),(\$\to\emptyword),\$\$\$\$,\emptyset]\\
r3 &=& [(\$\to\emptyword),\emptyset,1]
\end{eqnarray*}

In fact, the correctness of the last matrix can be seen by thinking of a strong version of MMMNF, i.e., of sMMMNF.
So, the idea is to reduce the number of nonterminals to~4, but ``paying'' for this by increasing all other parameters, leading to $\mathrm{SSCM}(4,2;4;3,3)=\mathrm{RE}$. Yet, we will present another (better) trade-off result  now, based on a different normal form. Compared to Theorem~\ref{thm:sscm-31522}, 
we use less nonterminals but more conditional matrices, and they are of length~3.
But one can also see that it really depends on what parameters are considered when different simulations are compared to each other. The parameters
$(4,2;4;3,3)$ of the sketch above are clearly worse than $(3,1;4;3,3)$ in terms of degree, but if one would count (e.g.) the overall number of rules involved in conditional matrices, then we would count 6 in the sketch, while we have 8 in the construction of Theorem~\ref{thm:sscm-31433}.
\end{remark}
\end{toappendix}

\noindent
The next result is a trade-off result to the previous theorem and also to \cite{MedKop2004}.

\begin{theorem}\label{thm:sscm-31433}
$\mathrm{SSCM}(3,1;4;3,3)=\mathrm{RE}$.     
\end{theorem}

\begin{proof}
Let $L \in \mathrm{RE}$ be generated by a grammar in $(3,2)$-GNF of the form $G= (V, T, P\cup \{AA\to \emptyword, BBB\to \emptyword\}, S)$ such that $P$ contains only context-free rules and $V\setminus T=\{S, A, B\}$ (see Propos.~\ref{propos:(3,2)-GNF}). 
Next, we define the $\mathrm{SSCM}$-grammar $G'= (V', T, P'\cup P'', S)$, where $V' = V \cup \{\#\}$ (assuming that $\# \not\in V$), $P'$ contains the (single-rule) unconditional matrices of the form $[(S\to \alpha), \emptyset,\emptyset]$ whenever $S\to \alpha \in P$ and $P''$ contains the three (multi-rule) matrices 
\begin{eqnarray*}
r1 &=& [(B\to\#),(B\to\#),(B\to\#),\emptyset,\#]\\
r2 &=& [(A\to\#),(A\to\#\#),\emptyset,\#]\\
r3 &=& [(\#\to\lambda),(\#\to\lambda),(\#\to\lambda),\#\#\#,\emptyset]
\end{eqnarray*}
\SV{For the formal correctness proof, we refer to the long version.}
\begin{toappendix}
\SV{\subsection{Correctness Proof for the Construction of Theorem~\ref{thm:sscm-31433}}}
The intended simulation of the two non-context-free deletion rules is obvious and hence $L(G)\subseteq L(G')$ is clear. 
Likewise, $G'$ satisfies the claims concerning the descriptional complexity parameters.

To prove $L(G)\supseteq L(G')$, we will actually show the following claim by induction on the length of the derivation of~$w\in (V')^*$:
If $S\Ra_{G'}^*w$, then either (1) $S\Ra_G^*w$ or (2) there exist three occurrences of $\#$ in $w$, i.e., $w=w_1\#w_2\#w_3\#w_4$ and there exists a sentential form $w'$ such that $S\Ra_G^*w'$ and either $w'=w_1Bw_2Bw_3Bw_4$ or $w'=w_1Aw_2w_3Aw_4$ and $w_2=\emptyword$ or, symmetrically, $w_3=\emptyword$.

The claim is trivially true if $w=S$ (derivation length 0).

Consider some sentential form $w$ such that $S\Ra_{G'}^n w$ for some $n>0$.
There must be some $v\in (V')^*$ such that $S\Ra_{G'}^{n-1}v \Ra_{G'}w$.
By induction, either  (1) $S\Ra_G^*v$ or (2) there exist three occurrences of $\#$ in $w$, i.e., $v=v_1\#v_2\#v_3\#v_4$ and there exists a sentential form $v'$ such that $S\Ra_G^*v'$ and either (a) $v'=v_1Bv_2Bv_3Bv_4$ or (b) $v'=v_1Av_2v_3Av_4$ and $v_2=\emptyword$ or, symmetrically, $v_3=\emptyword$.

Assume that (1) holds. Then, $\#$ does not occur in~$v$. If $S$ occurs in~$v$, then $v \Ra_{G'}w$ could be due to applying an unconditional matrix. This is clearly equivalent to applying some context-free rule in~$G$, i.e., $w$ is (also) a sentential form derivable in~$G$. Otherwise, we would have to apply either matrix $r1$ or matrix $r2$ (but not $r3$ by the absence of~$\#$).
On applying $r1$, we turn three occurrences of~$B$ into~$\#$, so that $w$ would satisfy condition (2a) above. On applying $r2$, one occurrence of~$A$ is turned into $\#$ and another one into $\#\#$, this way satisfying condition (2b) for~$w$.

Assume that (2) holds. If $S$ occurs in~$v$, then $v \Ra_{G'}w$ could be due to applying an unconditional matrix, basically applying the context-free rule $S\to\alpha$ on~$v$. First assume (2a) holds. Hence, $v=v_1\#v_2\#v_3\#v_4$ so that $v'=v_1Bv_2Bv_3Bv_4$ exists with $S\Ra^*G v'$ by induction. Applying $S\to\alpha$ on~$v$ means that there is a (unique) occurrence of~$S$ in some $v_i$ for $i\in\{1,2,3,4\}$, i.e., $v_i=v_i'Sv_i''$. Now, define $w_j=v_j$ if $j\in\{1,2,3,4\}\setminus \{i\}$ and $w_i=v_i'\alpha v_i''$. Then, $w=w_1\#w_2\#w_3\#w_4$ and, moreover, $v'\Ra_G w'$ with $w'=w_1Bw_2Bw_3Bw_4$.
The argument in case (2b) is analogous and hence omitted. In both cases, we see that condition (2) is satisfied for~$w$.

Now, assume that  (2) holds and that $v \Ra_{G'}w$ was not due to applying an unconditional matrix. As $\#$ is present in~$v$, then $r3$ must be applied.
This means that $\#\#\#$ is a substring of~$v$, so that, keeping up earlier index schemes, $v=v_1\#\#\#v_4$. On applying $r3$, we arrive at $w=v_1v_4$.
By induction, there is some $v'$ with $S\Ra_G^*v'$, such that either (2a) $v'=v_1BBBv_4$ or (2b) $v'=v_1AAv_4$. (Recall that we have these more restricted variants because $\#\#\#$ is a substring of~$v$.) Now, applying either (in (2a)) $BBB\to\emptyword$ or (in (2b)) $AA\to\emptyword$, we arrive at $w=v_1v_4$, i.e., $S\Ra_G^* w$ holds, which is condition (1) for~$w$.
\end{toappendix}
\qed 
\end{proof}

\section{Non-simple Semi-Conditional Matrix Grammars}

In the previous section, we arrived at some computational completeness results for simple SCM grammars. We only got one result with four nonterminals, but in this case, we could not get down to matrix length two. Therefore, we are now relaxing the simplicity condition, rather measuring non-simplicity as an additional parameter, to be able to get several results with four nonterminals and matrix length two. Several trade-offs will be observed.

\SV{\begin{table}[tb]
    \centering
    \begin{minipage}{.48\textwidth}
\small  \begin{eqnarray*}
r1 &=& [(B\to\#), (B\to\#),BBB,\#]\\
r2 & = & [(\#\to\emptyword),(B \to \#\#\#),A\#\#B,\#\#\#]\\
r3 &=& [(A\to\#),(A\to\#),AA,\#]\\
r4 &=& [(\#\to\lambda),(\#\to\lambda),\#\#\#\#,\#B] \\
r5 & = & [(\# \to \lambda),(\# \to \lambda), A\#\#A,BBB]\\
r6 & = & [(\# \to \lambda), (\# \to \lambda),B\#\#B,\emptyset] \\
r7 & = & [(\# \to \lambda), (\# \to \lambda),\#\#,A] 
\end{eqnarray*}
    \end{minipage} \quad\,   \begin{minipage}{.47\textwidth}
\small       \begin{eqnarray*}
r1 &=& [(B\to\#),(B\to AA), ABBBA,AA] \\
r2 & = & [(B \to \#\#\#),A\#AAB,\#\#] \\
r3 & = & [(\# \to \emptyword),(\# \to \emptyword), \#AA\#,\emptyset]\\
r4 &=& [(A\to\#),(A\to\#\#\#),AA,\#]\\
r5 &=& [(\#\to\emptyword),(\#\to\emptyword),\#\#\#\#,\emptyset]
\\
r6 & = & [(\# \to \emptyword), (\# \to \emptyword), B\#\#B,\emptyset]
\\
r7 & = & [(\# \to \emptyword), (\# \to \emptyword),\#\#,A]
\end{eqnarray*}
    \end{minipage}\\[1ex]
    \caption{Simulating $(3,2)$-GNF in two (related) ways: Theorem~\ref{thm:scm-434726} for $\mathrm{SCM}(4,3;4;7,2,6)$ on the left and Theorem~\ref{thm:scm-524724} for $\mathrm{SCM}(5,2;4;7,2,4)$ on the right.}
    \label{tab:two-simulations}
\end{table}}

\LV{\begin{table}[tb]
    \centering
    \begin{eqnarray*}
r1 &=& [(B\to\#), (B\to\#),BBB,\#]\\
r2 & = & [(\#\to\emptyword),(B \to \#\#\#),A\#\#B,\#\#\#]\\
r3 &=& [(A\to\#),(A\to\#),AA,\#]\\
r4 &=& [(\#\to\lambda),(\#\to\lambda),\#\#\#\#,\#B] \\
r5 & = & [(\# \to \lambda),(\# \to \lambda), A\#\#A,BBB]\\
r6 & = & [(\# \to \lambda), (\# \to \lambda),B\#\#B,\emptyset] \\
r7 & = & [(\# \to \lambda), (\# \to \lambda),\#\#,A] 
\end{eqnarray*}
    \caption{The simulation of non-context-free erasing rules of a grammar in $(3,2)$-GNF from Theorem~\ref{thm:scm-434726}.}
    \label{tab:scm-434726}
\end{table}}
\begin{theorem}\label{thm:scm-434726}
    $\mathrm{SCM}(4,3;4;7,2,6)=\mathrm{RE}$.
\end{theorem}
\begin{proof}
Let $L\in\mathrm{RE}$. 
We start with a type-0 grammar~$G$ in $(3,2)$-GNF that generates~$L$.
Hence, in Phase~1, when we only apply context-free rules which carry over as unconditional single-rule matrices, we can derive some string from 
$$\{ABB,AB\}^*\{S,\emptyword\}(\{BA,BBA\}\cup T)^*\subseteq L_{(3,2)}\,.$$
The two non-context-free erasing rules $BBB\to\emptyword$ and $AA\to\emptyword$ are \LV{meant to be }simulated by \LV{seven}\SV{7} conditional matrices; here, we employ a fourth nonterminal~$\#$, see Table~\SV{\ref{tab:two-simulations}}\LV{\ref{tab:scm-434726}}.

The intended simulation works as follows:
\begin{eqnarray*}
\alpha ABBBA \beta & \Ra_{r1} & \alpha A\#\#BA\beta \Ra_{r2} \alpha A\#\#\#\#A\beta \Ra_{r4} \alpha A\#\#A\beta \\
& \Ra_{r5} &  \alpha AA\beta \Ra_{r3} \alpha \#\#\beta \Ra_{r6/7} \alpha \beta
\end{eqnarray*}
Whether to apply matrices $r6$ or $r7$ in the last step depends on whether  $B$'s are ending~$\alpha$ and starting~$\beta$ or whether $\alpha=\beta=\emptyword$.

\SV{The formal correctness proof of this construction is in the long version.}
\begin{toappendix}
\SV{\subsection{Correctness Proof of the Construction of Theorem~\ref{thm:scm-434726}}}
Let us now collect some further observations.
Notice that to a string derivable in Phase~1, at best matrices $r1$ or $r3$ would be applicable due to the absence of $\#$, but the required permitting substrings will not occur in Phase~1. Hence, we consider a string $w=\alpha\zeta\beta$ as available in Phase~2, i.e., $\alpha\in\{AB,ABB\}^*$, $\zeta\in\{AA,ABBA,ABBBA,ABBBBA\}$, $\beta\in(\{BA,BBA\}\cup T)^*$, see Eq.~\eqref{eq:SF-GNF-(3,2)}. Only if $\zeta=AA$ or if $\zeta=ABBBA$, the two non-context-free erasing rules $BBB\to\emptyword$ and $AA\to\emptyword$ are applicable; how they are actually meant to be simulated is explained alongside the analysis of cases~1 and~4 below.

\smallskip\noindent
\underline{Case 1: $\zeta=AA$.} The only applicable matrix is now $r3$, which will produce a string $w_1$ that is obtained from~$w$ by turning any two occurrences of~$A$ into~$\#$. If $\alpha=\emptyword$ and $\beta\in T^*$, then (and only then) matrix~$r7$ would be applicable, correctly allowing to let $w_1=\#\# \beta$ derive the terminal string~$\beta$, this way simulating $AA\to\emptyword$. By the structure of our encoding, matrix $r5$ is not applicable, as this would mean that~$w$ has contained the substring $AAAA$. Matrix~$r6$ is only applicable if $\zeta$ has been turned into $\#\#$ within~$w_1$, so that we arrive at $w_2=\alpha\beta$, which is again correctly simulating $AA\to\emptyword$.

\smallskip\noindent
\underline{Case 2: $\zeta=ABBA$.} It is not hard to check that now, no matrix is applicable.

\smallskip\noindent
\underline{Case 3: $\zeta=ABBBA$.} Now, matrix~$r1$ is applicable (as intended) and will turn two occurrences of $B$ into~$\#$. Now, all potentially continuing matrices contain $\#\#$ in their permitting context checks, i.e., two $B$'s next to each other have been turned into~$\#\#$ when producing the resulting string~$w_1$. By the structure of our encoding and in the present case~3, we cannot see the substring $B\#\#B$ within~$w_1$, as this would require the substring~$B^4$ in~$w$. In order to get the substring $A\#\#A$ (for applying matrix $r5$), we have to touch two $B$ in $\alpha\beta$ and none in $\zeta$, so that the substring $BBB$ is maintained in~$w_1$, hence preventing us from applying~$r5$. Also, to a string obtained from $w$ without changing the center~$\zeta$, although $BBB$ would be still present, we cannot re-apply $r1$ due to the presence of~$\#$. 
Hence, we have to continue with matrix~$r2$, having arrived at $w_1=\alpha A \#\#BA\beta$ to pass the substring tests. Now, matrix $r2$ will delete one occurrence of $\#$ and turn any occurrence of $B$ into $\#\#\#$ to produce~$w_2$. As the reader can verify, now the only potentially applicable matrix is $r4$, which means that $w_2=\alpha A \#^4A\beta$. The string obtained from applying $r4$ is then $w_3=\alpha A \#\#A\beta$. Now, matrix $r5$ must be applied, leading to  $w_4=\alpha A A\beta$. This means that the deletion rule $BBB\to\emptyword$ was successfully simulated.

\smallskip\noindent
\underline{Case 4: $\zeta=ABBBBA$.}
As in case~3, we can argue that from~$w$, we must continue with $w_1=\alpha A\#\#BBA\beta$ or with $w_1'=\alpha AB \#\# BA\beta$. With $w_1$, we might still get $w_2=\alpha A\#\#\#\#BA\beta$, but now the derivation is stuck, because $r4$ has the forbidden context $\#B$. From $w_1'$, we can get $w_2'=\alpha ABBA\beta$ by using matrix $r6$, but now there is no continuation as in case~2. Hence, we cannot successfully derive any terminal string in this case.

Now, by a straightforward formal induction argument, one can see that the languages of~$G$ and of~$G'$ are equal.
\end{toappendix}
\qed    
\end{proof}

Now we get a trade-off result, as the construction above has uncomparable degree but less non-simple matrices compared to Theorem~\ref{thm:scm-524724}. %\LV{, and even more compared to Theorem~\ref{thm:scm-634723}}.

\begin{theorem}\label{thm:scm-524724}
$\mathrm{SCM}(5,2;4;7,2,4)=\mathrm{RE}$.
\end{theorem}

\begin{proof}
Let $L\in\mathrm{RE}$. 
We again start with a type-0 grammar~$G$ in $(3,2)$-GNF that generates~$L$.
Compared to the previous construction, only the matrices simulating the two non-context-free erasing rules $BBB\to\emptyword$ and $AA\to\emptyword$ change\LV{ into the following ones}\SV{ as can be seen in Table~\ref{tab:two-simulations}}.
\LV{\begin{eqnarray*}
r1 &=& [(B\to\#),(B\to AA), ABBBA,AA] \\
r2 & = & [(B \to \#\#\#),A\#AAB,\#\#] \\
r3 & = & [(\# \to \emptyword),(\# \to \emptyword), \#AA\#,\emptyset]\\
r4 &=& [(A\to\#),(A\to\#\#\#),AA,\#]\\
r5 &=& [(\#\to\emptyword),(\#\to\emptyword),\#\#\#\#,\emptyset]
\\
r6 & = & [(\# \to \emptyword), (\# \to \emptyword), B\#\#B,\emptyset]
\\
r7 & = & [(\# \to \emptyword), (\# \to \emptyword),\#\#,A]
\end{eqnarray*}}
The intended simulation works as follows:
\begin{eqnarray*}
\alpha ABBBA \beta & \Ra_{r1} & \alpha A\#AABA\beta \Ra_{r2} \alpha A\#AA\#\#\#A\beta \Ra_{r3} \alpha \#AA\#\beta \\
& \Ra_{r3} &  \alpha AA\beta \Ra_{4} \alpha \#\#\#\#\beta \Ra_{r5}\alpha \#\#\beta  \Ra_{r6/7} \alpha \beta
\end{eqnarray*}
Whether to apply matrices $r6$ or $r7$ in the last step depends on whether  $B$'s are ending~$\alpha$ and starting~$\beta$ or whether $\alpha=\beta=\emptyword$.
\SV{The formal correctness proof of this construction is in the long version.}
\begin{toappendix}
\SV{\subsection{Correctness Proof of the Construction of Theorem~\ref{thm:scm-524724}}}

Notice that to a string derivable in Phase~1, at best matrices $r1$ or $r4$ would be applicable due to the absence of $\#$, but the required permitting substrings will not occur in Phase~1. Hence, we consider a string $w=\alpha\zeta\beta$ as available in Phase~2, i.e., $\alpha\in\{AB,ABB\}^*$, $\zeta\in\{AA,ABBA,ABBBA,ABBBBA\}$, $\beta\in(\{BA,BBA\}\cup T)^*$. Only if $\zeta=AA$ or if $\zeta=ABBBA$, the two non-context-free erasing rules $BBB\to\emptyword$ and $AA\to\emptyword$ are applicable; how they are actually meant to be simulated is explained alongside the analysis of cases~1 and~4 below.

\smallskip\noindent
\underline{Case 1: $\zeta=AA$.} The only applicable matrix is now $r4$, which will produce a string $w_1$ that is obtained from~$w$ by turning any two occurrences of~$A$ into~$\#$ and $\#\#\#$, respectively. The presence of the substring $\#\#\#$ in~$w_1$ directly blocks applying $r2$ and $r4$. The substring $\#\#$ (but not $\#\#\#$) is required to be present for applying matrix%ces $r3$ or
~$r6$. By the structure of strings in Phase~2, the presence of substring $AA$ in~$w$ implies the absence of $BBB$ in~$w$ and hence in $w_1$, so that matrix~$r1$ is also inapplicable. If matrix~$r3$ would be applicable now, this would mean that~$w$ had to contain $AAAA$, which is impossible.
Hence, only $r5$ or $r7$ might apply. In both cases, actually $\zeta$ must have been replaced by $\#\#\#\#$ when producing~$w_1$, so that we know that the resulting string~$w_2$ equals $\alpha \#\#\beta$. In order to continue with~$w_2$, we have to apply $r6$ or $r7$, differentiating between a further continuation or a possible termination of the derivation process. 

\smallskip\noindent
\underline{Case 2: $\zeta=ABBA$.} It is not hard to check that now, no matrix is applicable.

\smallskip\noindent
\underline{Case 3: $\zeta=ABBBA$.} By the structure of sentential forms in Phase~2, only matrix~$r1$ is applicable and will turn one occurrence of~$B$ into~$\#$ and another one into~$AA$, this way leading to the string~$w_1$. As there is (only) one occurrence of~$\#$, as well as one of~$AA$, within~$w_1$, now only matrix~$r2$ can be applicable. This is only possible if, within the center~$\zeta$, the first occurrence of~$B$ was turned into~$\#$ and the second one into~$AA$, i.e., if $w_1=\alpha A\#AABA\beta$. In the resulting string~$w_2$, compared to~$w_1$, %the only occurrence of~$\#$ was deleted and 
any occurrence of $B$ has been turned into~$\#\#\#$. In the presence of the substring~$\#\#\#$ within~$w_2$, but not $\#^4$, only matrix~$r3$ %ces $r3$, $r6$ and~$r7$ 
might be applicable. 
As $\alpha$ ends with a~$B$, the permitting context $\#AA\#$ implies that $w_2=\alpha A\#AA\#\#\#A\beta$, leading to $w_3=\alpha A\#AA\#A\beta$ and then $w_4=\alpha AAAA\beta$ deterministically by similar arguments.
Now, the reasoning presented in case~1 is by large applicable up to the conclusion that on~$w_4$, matrix~$r4$ must be applied, leading to a string~$w_5$ where two neighboring occurrences of~$A$ have been turned into $\#\#\#\#$ so that (1) $w_5=\alpha \#\#\#\#AA\beta$ or  (2) $w_5=\alpha A\#\#\#\#A\beta$ or (3) $w_5=\alpha AA\#\#\#\#\beta$. The alternative is to convert two occurrences of $A$ at distance~2, leading to (4) $w_5=\alpha \# AA\#\#\#\beta$ or (5) $w_5=\alpha \#\#\# AA\#\beta$; with (4) or (5), the derivation could continue with using matrix~$r3$ twice as explained above, leading to  $w_7=\alpha AA\beta$.  
In cases (1)-(3), only matrix~$r5$ is applicable, leading to  $w_6=\alpha \#\#AA\beta$ or  $w_6=\alpha A\#\#A\beta$ or $w_6=\alpha AA\#\#\beta$. However, now the derivation is stuck.
In conclusion, string~$w=\alpha ABBBA\beta$ will be necessarily converted into $w_7$ which correctly simulates the erasing rule $BBB\to\emptyword$.

\smallskip\noindent
\underline{Case 4: $\zeta=ABBBBA$.} Again, no matrix is applicable.
Now, by a straightforward formal induction argument, one can see that the languages of~$G$ and of~$G'$ are equal.
\end{toappendix}
\qed    
\end{proof}

\noindent
We are now presenting yet another trade-off result, now based on sMMNF.

\SV{\begin{table}[tb]
    \centering
    \begin{minipage}{.48\textwidth}
\small \begin{eqnarray*}
r1 &=& [(S\to \$\$),0S,\emptyset]\\
r2 &=& [(0\to \$\$\$),(0\to \$),0\$\$0,\$\$\$]\\
r3 &=& [(1\to \$\$\$),(1\to \$),1\$\$1,\$\$\$]\\
r4 &=& [(\$\to\emptyword),(\$\to\emptyword),\$^6,\emptyset]\\
r5 &=& [(\$\to\emptyword),(\$\to\emptyword),0\$^40,\emptyset]\\
r6 &=& [(\$\to\emptyword),(\$\to\emptyword),1\$^41,0\$]\\
r7 &=& [(\$\to\emptyword),(\$\to\emptyword),\emptyset,1]
\end{eqnarray*}
    \end{minipage} 
    \quad\,   
\begin{minipage}{.47\textwidth}
\small   \begin{eqnarray*}
r1 &=& [(S\to SS),0S,SS]\\
r2 &=& [(0\to SSS),(0\to S),0SS0,SSS]\\
r3 &=& [(1\to SSS),(1\to S),1SS1,SSS]\\
r4 &=& [(S\to\emptyword),(S\to\emptyword),S^6,\emptyset]\\
r5 &=& [(S\to\emptyword),(S\to\emptyword),0S^40,\emptyset]\\
r6 &=& [(S\to\emptyword),(S\to\emptyword),1S^41,0S]\\
r7 &=& [(S\to\emptyword),(S\to\emptyword),\emptyset,1]
\end{eqnarray*}
    \end{minipage}\\[1ex]
    \caption{Simulating sMMNF in two (related) ways: Theorem~\ref{thm:scm-634723} for $\mathrm{SCM}(6,3;4;7,2,3)$ on the left and Theorem~\ref{thm:scm-633-24} for $\mathrm{SCM}(6,3;3;*,2,4)$ on the right.}
    \label{tab:another-two-simulations}
\end{table}}

\begin{theorem}\label{thm:scm-634723}
$\mathrm{SCM}(6,3;4;7,2,3)=\mathrm{RE}$.
\end{theorem}
\begin{proof} We start with an sMMNF grammar~$G$ for a given arbitrary RE language $L\subseteq T^*$, i.e., $G$ is a type-0 grammar of the form $G= (V, T, P\cup \{0\$0\to \$, 1\$1\to \$,\$\to\emptyword\}, S)$ with $L(G)=L$ such that $P$ contains only context-free rules and $N\coloneqq V\setminus T=\{S, 0, 1,\$\}$.
We suggest the\LV{ following} simulation rules for $0\$0\to\emptyword$ and for $1\$1\to\emptyword$ in the SCM grammar $G'$\SV{ as given in Table~\ref{tab:another-two-simulations}}.
\LV{
\begin{eqnarray*}
r1 &=& [(S\to \$\$),0S,\emptyset]\\
r2 &=& [(0\to \$\$\$),(0\to \$),0\$\$0,\$\$\$]\\
r3 &=& [(1\to \$\$\$),(1\to \$),1\$\$1,\$\$\$]\\
r4 &=& [(\$\to\emptyword),(\$\to\emptyword),\$^6,\emptyset]\\
r5 &=& [(\$\to\emptyword),(\$\to\emptyword),0\$^40,\emptyset]\\
r6 &=& [(\$\to\emptyword),(\$\to\emptyword),1\$^41,0\$]\\
r7 &=& [(\$\to\emptyword),(\$\to\emptyword),\emptyset,1]
\end{eqnarray*}}

The context-free rules of~$G$ are simulated by unconditional matrices in~$G'$ except for (a) the rules $S\to u\$v$ in~$G$ that will be simulated by $S\to uSv$ and then $r1$, as well as for (b) $\$\to\emptyword$ in~$G$ that is simulated by $r7$.
We briefly explain the idea behind.
Notice that for each rule $S\to u\$v$ in~$G$, there is also a rule $S\to uSv$ in~$G$. The reason why Geffert (and hence also MM) normal forms have this specific feature (and not, as one might expect, directly the rule $S\to\$$) is that otherwise, $\emptyword\in L(G)$ would be always true. Another way to avoid the trivial derivation $S\Ra \$ \Ra \emptyword$ in an MMNF grammar would be to check, when  $S\to\$$ is applied, that indeed some of the rules $S\to uSv$ have been applied before. This is the purpose of the permitting context $0S$ in matrix~$r1$. Therefore, we find $S\Ra_G^*\alpha S\beta\Ra_G\alpha \$\beta$ if and only if $S\Ra_{G'}\alpha S\beta\Ra_{G'}\alpha \$\$\beta$. 
(b) $\$\to\emptyword$ is meant to be the very last rule applied in~$G$; in $r7$, this is explicitly checked with the help of the forbidden context~1. This check makes use of the fact that we started with a grammar~$G$ in strong MMNF.\LV{ We come back to this issue later.

Notice that $G'$ will use no additional nonterminals compared to~$G$. Also, the claimed descriptional complexity measure of~$G'$ is easily verified.

}
We now present intended derivations.
\begin{eqnarray*}
\alpha 0\$\$0\beta &\Ra_{r2}&\alpha \$^6\beta \Ra_{r4}\alpha \$^4\beta \Ra_{r5/r6}\alpha \$\$\beta\,,\\ 
\alpha 1\$\$1\beta &\Ra_{r3}&\alpha \$^6\beta \Ra_{r4}\alpha \$^4\beta \Ra_{r5/r6/r7}\alpha \$\$\beta\,, \end{eqnarray*} 
Which matrix to apply in the last step depends on the prefix~$\alpha$ and suffix~$\beta$. 
\SV{The formal correctness proof of this construction is in the long version.}
\begin{toappendix}
\SV{\subsection{Correctness Proof of the Construction of Theorem~\ref{thm:scm-634723}}}

As can be seen from Eq.~\eqref{eq:SF-sMMNF}, a `typical word' in Phase~2 is of the form $w=\alpha \$\beta$, with $\alpha\in \{10,100\}^*\{1,10,\emptyword\}$ and $\beta\in \{1,01,\emptyword\}(\{01,001\}\cup T)^*$. Instead, in the simulating grammar~$G'$, we have $w_0=\alpha \$\$\beta$. The `good cases' that we expect in the center are the substrings $0\$(\$)0$ and  $1\$(\$)1$. Both cases are treated nearly completely symmetrically. Hence consider $w_0=\alpha' 0\$\$0\beta'$ in what follows. By the structure of the encoding of strong MMNF, $\alpha'$ is not empty, as it starts with a~1, and similarly, $\beta'\neq\emptyword$, as it ends with a~1. Hence, $r7$ is not applicable.
(This is the one place where the argument differs with the center   $1\$(\$)1$.) Hence, we have to apply matrix $r2$. This yields~$w_1$, being different from $w$ in the sense that one occurrence of~0 was turned into~$\$$ and another one into~$\$\$\$$. Now, we cannot apply $r5$, as this would mean that in~$w$, we find four zeros in a row (observe that $\$\$\$\$$ can only encode two occurrences of~0 next to each other). We might try to apply $r6$, as~$w$ can contain a substring $1001$ somewhere. However, then the center still contains $0\$$, so that in fact $r6$ is inapplicable. (This is the other place where the argument differs with the center   $1\$(\$)1$, because $w_0$ cannot contain the substring $11$ at all, so that the forbidden context for $r5$ is not necessary.) Therefore, $r4$ has to be applied, which means that $w_1=\alpha' \$\$\$\$\$\$\beta'$. The resulting string is $w_2=\alpha' \$\$\$\$\beta'$. Now, applying $r5$ or $r6$, depending on the suffix of $\alpha'$ and the suffix of $\beta'$, is enforced.
(As the strings $\alpha'$ and $\beta'$ could be empty for the center  $1\$\$1$ in $w_0$, we might apply $r7$ instead then.)
This shows that the erasing non-context-free rules are correctly simulated. 

Recall that in order to terminate, one can apply $r7$ at the end, which simulates the context-free erasing rule $\$\to\emptyword$ of the MMNF grammar.

Straightforward induction shows that the languages of the original sMMNF grammar~$G$ and of the designed SCM grammar~$G'$ coincide.
\end{toappendix}
\qed
\end{proof}
\begin{remark}
The reader might have wondered if one could not save one of the matrices by opting for a strong version of MMMNF instead of sMMNF. However, then the forbidden context check in $r6$ seems to be impossible. Therefore, this idea would not work.
\end{remark}

\section{\LV{Computational Complete SCM Grammars with only Three Nonterminals}\SV{SCM Grammars: When Three Nonterminals Suffice}}

In %the last part of 
this section, we ask ourselves if we could try to obtain computational completeness results with three nonterminals only. Indeed, we can achieve this, as we show in the following, but at the expense of having an unbounded number of (simple) conditional matrices. Otherwise, the idea is simply to modify the previous construction and to re-use the nonterminal~$S$ also in Phase~2. More precisely, the role of the central $\$$ is now taken by the substring $SS$ whose absence must hence be checked in each simulation of a context-free rule. A bit surprisingly, this idea does not increase the degree of the simulating grammar in comparison with the previous Theorem~\ref{thm:scm-634723}.

\SV{It is interesting to compare this result with that on (classical) matrix grammars where it is known that three nonterminals suffice to characterize $\mathrm{RE}$, see~\cite{Feretal07} where the lengths of the matrices are arbitrarily big.}
\LV{These results should also be compared with the following two nonterminal complexity results:
\begin{itemize}
    \item For matrix grammars, it is known that three nonterminals suffice to characterize $\mathrm{RE}$, see \cite{Feretal07}.
    \item For SSCG, again three  nonterminals suffice to characterize $\mathrm{RE}$.\footnote{Unpublished result, by a group of researchers from Singapore, as well as from Germany and New Zealand.}
\end{itemize}}

\begin{theorem}\label{thm:scm-633-24}
$\mathrm{SCM}(6,3;3;*,2,4)=\mathrm{RE}$.
\end{theorem}

\begin{proof} 
We again start with a grammar~$G$ in sMMNF that generates an arbitrary RE language  $L\subseteq T^*$. This means that $G$
a type-0 grammar of the form $G= (V, T, P\cup \{0\$0\to \$, 1\$1\to \$,\$\to\emptyword\}, S)$ with $L(G)=L$ such that $P$ contains only context-free rules and $N\coloneqq V\setminus T=\{S, 0, 1,\$\}$.
Any context-free rule $S\to\gamma$ of~$G$ is simulated by the simple conditional matrix $r_\gamma=[(S \to \gamma), \emptyset, SS]$. 
The simulation rules for $0\$0\to\emptyword$ and for $1\$1\to\emptyword$ in the SCM grammar~$G'$ are presented in \SV{Table~\ref{tab:another-two-simulations}}\LV{the following}.
\LV{
\begin{eqnarray*}
r1 &=& [(S\to SS),0S,SS]\\
r2 &=& [(0\to SSS),(0\to S),0SS0,SSS]\\
r3 &=& [(1\to SSS),(1\to S),1SS1,SSS]\\
r4 &=& [(S\to\emptyword),(S\to\emptyword),S^6,\emptyset]\\
r5 &=& [(S\to\emptyword),(S\to\emptyword),0S^40,\emptyset]\\
r6 &=& [(S\to\emptyword),(S\to\emptyword),1S^41,0S]\\
r7 &=& [(S\to\emptyword),(S\to\emptyword),\emptyset,1]\\
\end{eqnarray*}}
None of the matrices $r2$ through $r7$ can be applied already when simulating context-free rules, because each matrix requires the presence of (at least) two occurrences of~$S$ one way or the other. Therefore, the arguments from the proof of Theorem~\ref{thm:scm-634723} nearly literally translate to this case.  
\qed
\end{proof}

\SV{\noindent}
We  now present another trade-off result\SV{, with uncomparable degrees}\LV{: while the permitting degree increases, but the forbidden string length decreases, as well as the number of non-simple conditional matrices}.
\begin{theorem}\label{thm:scm-723-23}
    $\mathrm{SCM}(7,2;3;*,2,3)=\mathrm{RE}$.
\end{theorem}

\begin{proof}
Once more, we start with a grammar~$G$ in sMMNF that generates an arbitrary given RE language  $L\subseteq T^*$. This means that $G$
a type-0 grammar of the form $G= (V, T, P\cup \{0\$0\to \$, 1\$1\to \$,\$\to\emptyword\}, S)$ with $L(G)=L$ such that $P$ contains only context-free rules and $N\coloneqq V\setminus T=\{S, 0, 1,\$\}$.
Any context-free rule $S\to\gamma$ of~$G$ is simulated by the simple conditional matrix $r_\gamma=[(S \to \gamma), \emptyset, SS]$. 
We suggest the following simulation rules for $0\$0\to\emptyword$ and for $1\$1\to\emptyword$ in the SCM grammar $G'$, with $\eta\in\{0,1\}$.

\begin{center}
\begin{minipage}{.47\textwidth}
\small
\begin{eqnarray*}
r1 &=& [(S\to S1S),0S0,1S]\\
r2 &=& [(0\to 11),(0\to SSSS11),0S1S0,11]\\
r3 &=& [(1\to S),(1\to SSS),1S1S1,SS]\\
r4 &=& [(S\to\emptyword),(S\to\emptyword),SS1SSSS,\emptyset]\end{eqnarray*}
\end{minipage}\quad 
\begin{minipage}{.47\textwidth}
\small
\begin{eqnarray*}
r5_\eta &=& [(S\to\emptyword),(S\to\emptyword),\eta 1S1SSS,\emptyset]\\
%r5' &=& [(S\to\emptyword),(S\to\emptyword),11S1SSS,\emptyset]\\
r6 &=& [(S\to\emptyword),(S\to\emptyword),0S1SSS0,\emptyset]\\
r7 &=& [(S\to\emptyword),(S\to\emptyword),\emptyset,0]\\
r8 &=& [(1\to\emptyword),\emptyset,S]
\end{eqnarray*}
\end{minipage}
\end{center}
The intended simulations are:
$$\alpha 0S1S0 \beta \Ra_{r2} \alpha 11S1SSSSS11\beta \Ra_{r5_1}^2\alpha 11S1S11 \beta, $$
$$\alpha 1S1S1 \beta \Ra_{r3} \alpha SS1SSSS\beta \Ra_{r4} \alpha S1SSS\beta \Ra_{r5_0/6}\alpha S1S \beta, $$
(with $\alpha,\beta\neq\emptyword$, with $\alpha$ ending with the same symbol as $\beta$ starts) and
$$1S1S1 t \Ra_{r7}111t\Ra_{r8}^3t\text{ for } t\in T^*\,. $$
From these three cases, $L(G)\subseteq L(G')$ follows easily by induction.
\SV{The converse direction is unfortunately a quite intricate case analysis within an inductive argument and is hence deferred to the long version of this paper.}
\begin{toappendix}
\SV{\subsection{Proof of the Inlcusion $L(G')\subseteq L(G)$ in Theorem~\ref{thm:scm-723-23}}}
Conversely, consider some sentential form $w$ that we assume to be derivable by $G$, i.e., $S\Ra_G^* w$, such that there is some sentential form $w'$ derivable by $G'$, i.e.,  $S\Ra_{G'}^* w'$, where $w'=w$ if $w$ contains an occurrence of~$S$ (and then exactly one), while $w'=\alpha S1S\beta$ if $w=\alpha \$\beta$.
We are going to prove (inductively) that any sentential form~$v'$ derivable from $w'$ in $G'$ and being of the form $v'=\delta'S1S\eta'$ for some $\delta'\in\{0,1\}^*$ and $\eta'\in (T\cup\{0,1\})^*$ either has a counterpart~$v$ as described above for the pair $(w,w')$ such that $S\Ra^*_Gv$, or $w'\Ra_{G'}t$ is invalid for any terminal string~$t\in T^*$. This claim then proves that $L(G')\subseteq L(G)$ as $S$ itself gives the induction basis for this claim.

Next, consider $w'=w=\alpha S\beta\in L_{\text{sMM}}$, i.e., $\alpha\in \{100,10\}^*$, $\beta\in \{001,01\}^*$ according to Eq.~\eqref{eq:SF-sMMNF}.
Hence, neither $1S$ nor $S1$ are substrings of $w'$, enabling possible derivations using $r_\gamma=[(S \to \gamma), \emptyset, S1]$ or $r1$. All other matrices either require at least two occurrences of~$S$ or none, so they cannot be applied on~$w'$. As argued in the proof of Theorem~\ref{thm:scm-634723}, we can assume $\gamma=u_1Su_2$ for some $u_1\in\{100,10\}^+$ and $u_2\in T\cup \{001,01\}^*$. Now if $w'\Ra_{r_\gamma}v'$, then choose $v=v'$ and the context-free rule $S\to\gamma$ of~$G$ to see why $w\Ra_G v$.
If $\alpha=\emptyword$, this is in fact the only possibility.
If $\alpha\neq \emptyword$, then $r1$ might be applicable if $\beta\neq\emptyword$. Namely, by the structure of our encoding, $\alpha$ will end with~0 and $\beta$ can start with~0. The resulting string $v'=\alpha S1S\beta$ corresponds to the string $v=\alpha \$\beta $. Now, as $\alpha\neq\emptyword$, $w$ was obtained from a sentential form $x$ by applying a rule $S\to u_1'Su_2'$. With this rule, there is also a rule $S\to u_1'\$u_2' $ in~$G$. When we apply this on~$x$, we obtain~$v$, hence proving the claim in this case.

Now, consider $w'=\alpha S1S \beta$ derivable in~$G'$ such that we have a $w=\alpha\$ \beta$ derivable in~$G$.
This means that $\alpha$ is a word over $\{0,1\}$ and $\beta$ is a word over $T\cup\{0,1\}$. By the presence of the substring $S1S$, none of the simple conditional matrices of~$G'$ of length~1 is applicable. 

If $\alpha\in\{1\}^*$ and $\beta\in(\{1\}\cup T)^*$, then (as $w\in L_{\text{sMM}}$) $\alpha=\beta=\emptyword$ or $\alpha=1$ and $\beta\in\{1\}T^*$.
We could apply $r7$, followed by $r8$ (possibly multiple times) to arrive at a terminal string (suffix of~$\beta$) as intended.
If $\alpha=\beta=1$ (or more generally, $\beta$ starts with~1, i.e, $\beta=1\beta'$), then we could also apply $r3$, leading to 
(a) $v'=SSSSSS1\beta'$, (b) $v'= SSS1SSS\beta'$, or (c) $v'=1SSSSSS\beta'$.
In cases (a) and (b), no matrix is applicable, the derivation is stuck. In case (c), only the matrices $r_\gamma$ are applicable. If any other but the first occurrence of $S$ is used for some $r_\gamma$, then the derivation will be stuck again, because $\gamma$ always starts with a~1, so that we get the substring $S1$. After unintendedly restarting the simulation on the first~$S$, we can then (at some point) start deriving something from the second~$S$ etc.
But notice that any sentential form derived in this way will start with the prefix~$11$.
At some point, to make any progress, we have to apply $r1$. As this checks for the permitting context $0S0$, we know that next, we have to apply $r2$, which is impossible due to the forbidden context~$11$.
Hence, unintended re-starts are not possible.

After having looked into this very particular case, consider $\alpha=\alpha'1$ with $\alpha'\neq\emptyword$. This means that $\alpha$ contains some occurrences of~0, so that $r7$ is not applicable. As $w\in L_{\text{sMM}}$, we can assume that $\alpha'=\alpha''0$. The only applicable matrix is now $r3$, enforcing $\beta=1\beta'$ by the permitting context. Hence, $w'\Ra_{r3}w_1'$ by 
replacing one occurrence of~1 in $w'$ by $S$ and another one by~$SSS$.
The permitting context conditions of the matrices allow only one possible continuation, namely with $r4$. 
If $r4$ is applicable, then $w_1'=\alpha''0SS1SSSS\beta'$. Again in particular due to the permitting context conditions, with $w_1'\Ra_{G'} w_2'$, we know that we have to delete one occurrence of~$S$ per $S$-block in $w_1'$ when applying~$r4$ in order to continue the derivation. Hence, $w_2'=\alpha''0S1SSS\beta'$. Now, $r6$ is enforced, so that $\beta'=0\beta''$ is clear. If we now restart on $\alpha''01SS0\beta''$, then with arguments similar to our previous considerations, we see that the derivation will get stuck.
Therefore,  $w_3'=\alpha''0S1S0\beta''$ is enforced. Now, observe that from~$w'$ to~$w_3'$, the central $1S1S1$ was changed into $S1S$. If we now look at~$v$ obtained from~$w$ by applying $1\$1\to\$ $, then the counterpart of~$v$ within $G'$ is  obviously $v'=w_3'.$ This shows the inductive claim in this case.

Finally, reconsider $w'=\alpha S1S \beta$ derivable in~$G'$ such that we have a $w=\alpha\$ \beta$ derivable in~$G$. According to the previous case analysis, we can assume that $\alpha=\alpha'0$. The context conditions leave $r2$ as the only possibility of any matrix to apply on~$w'$. Hence, $\beta=0\beta'$. We have to replace one occurrence of~0 by $11$ and another one by the long string $SSSS11$, yielding $w_1'$. How can we ever continue with this long substring?
Some case analysis shows that we have to apply $r5_1$ now. Even more, we know that $w_1'=\alpha'11S1S SSSS11\beta'$. Now, deleting two occurrences of $S$ could lead to $\alpha'111 SSSS11\beta'$ and possibly allow a re-start as the only way to continue. However, the substring $S1$ immediately also blocks this attempt. Hence, the next string in the derivation must be  $w_2'=\alpha'11S1SSS11\beta'$. With similar considerations, we must re-apply $r5_1$ to get $w_3'=\alpha'11S1S11\beta'$.\footnote{At this point, the reader may ask why we did not use the rule $0\to SS11$ in matrix $r3$. But then, we could use $0\to 11$ at a wrong place and still continue with $r6$.}
 This string $w_3'$ does not correspond to any string derivable in~$G$ as above.
 Rather, $w_3'$ can be thought of being derived from $w'$ by replacing $0S1S0$ by $11S1S11$. Yet, the idea should be clear: we can now continue ``as above''.
 Yet, there some differences in this argument, so that we rather repeat it here to make the reader aware of the arguments.
 First, it might well be that $0\notin sub(w_3')$, so that $r7$ is applicable on~$w_3'$, followed by seven applications of $r8$ to derive a terminal word~$t$.
 What does this mean in~$G$? Clearly, $w=10\$01 t$ must then be true, as only then $w_3'=111S1S111t$. But this is clearly perfectly fine. 
 But in any case, we can also apply $r3$ to $w_3'$ and then, arguing by cases, we should apply $r4$, so that
 $$w_3'=\alpha'11S1S11\beta'\Ra_{r3}w_4'=\alpha'1SS1SSSS1\beta'\Ra_{r4}w_5'=\alpha'1S1SSS1\beta'\,.$$
Now, by the encoding that has led to~$w$ within~$G$, we know that $\alpha'\neq\emptyword$, i.e., it ends either with~$0$ or with~$1$. 
Depending on the case, $r5_0$ or $r5_1$ have to be applied next to yield 
$w_6'=\alpha'1S1S1\beta'$, again with short discussions about impossible restarts should we (instead) delete the leftmost occurrence of~$S$ in~$w_5'$. 

Now, we have to discuss the `next round'. But, depending on whether $\alpha'$ ends with~$0$ or~$1$ we either are in one of the two cases we already discussed, or we detect that $\alpha'$ ends on a symbol that is not a prefix of $\beta'$ and get stuck. But in such a case, also the derivation starting from~$w$ in~$G$ would be stuck. If the last letter in~$\alpha'$ matches the first one of $\beta'$, we will see the derivation $w_6'\Ra_{G'} w_7'\Ra_{G'}w_8'\Ra_{G'}w_9'=\alpha'S1S\beta'$. Clearly, $w_9'$ corrsponds to $v=\alpha'\$\beta' $ in~$G$, which is derivable from $w$ in one step using $0\$0\to\$ $.

Inductively, this reasoning shows $L(G')\subseteq L(G)$.
\end{toappendix}
\qed 
\end{proof}
If one likes to highlight this, we can also deduce from Theorem~\ref{thm:scm-633-24} that is actually slightly better than Theorem~\ref{thm:scm-723-23} in the ``long-matrices-count'':

\begin{corollary}
For each RE language~$L$, there is a SCM grammar generating~$L$ with only three nonterminals and only six matrices of length two; all other matrices have length~1.
\end{corollary}

Our last result concerning the theme telling that sometimes, three nonterminals are enough to generate all RE languages, is the following one. Its proof is significantly different from the ones previously presented in this paper, as it uses a simulation of graph-controlled grammars with only two nonterminals $A$ and~$B$\SV{ (not formally introduced in this paper so far). Thus, we the whole argument can be found in the long version}. The argument itself is a non-trivial adaptation of the proof of Cor.~6 in~\cite{Feretal07}. More precisely, in the simulation of the given graph-controlled grammar, we use a third nonterminal~$C$ to encode the current vertex (state, viewed as a number) of the control graph explicitly in unary in the sentential form. We can use a (long) sequence of rules $C\to\emptyword$ in a matrix to lower-bound the state number and the (long) forbidden context together to make sure that a state transition together with a successful rule application are properly simulated. The failure case is more tricky, as we have to split the tests (absence of left-hand side and upper bound on state number) over two matrices.
\begin{toappendix}
\SV{\subsection{An Argument Based on Graph-controlled Grammars}}
We now describe graph-controlled grammars slightly less formal than other grammars, rather from a graph-theoretic perspective.
Recall that a graph-controlled grammar $G$ possesses an underlying control graph with bi-colored directed edges (green and red). Alternatively, this can be viewed as a directed signed graph.
A context-free rule associated to each node of this graph. When processing a sentential form, we also have to store the current node in a configuration of the grammar. So, if $w$ is the current string and the processing is in node $n$, with associated rule $Y\to \alpha$, there are two cases that could happen:
\begin{itemize}
    \item either $Y$ occurs in $w$, in which case an occurrence of $Y$ in $w$ is replaced  by $\alpha$, yielding a new sentential form $w'$; moreover, we choose a green arc leading from $n$ to some node $n'$, so that $(n',w')$ is the new configuration of the grammar;
    \item or $Y$ does not occur in $w$, in which case $w$ is not changed, i.e., $w'=w$, but we choose a red arc leading from $n$ to some node $n'$, so that $(n',w')$ is the new configuration of the grammar.
\end{itemize}

For simplicity, rules are written as $(\ell:Y\to \alpha,\sigma(\ell),\varphi(\ell))$,
where $\ell$ is the node name, $Y\to \alpha$ is the context-free rule associated to~$\ell$, $\sigma(\ell)$ collects all nodes that can be reached by using green arcs from~$\ell$, while $\varphi(\ell)$ collects all nodes that can be reached by using red arcs from~$\ell$.

Moreover, there is a designated initial node~$\mathfrak{i}$ and a designated final node~$\mathfrak{f}$;
with $S$ being the start symbol, $(S,\mathfrak{i})$ is the initial configuration, while $(w,\mathfrak{f})$ is a final configuration if $w$ is a terminal string. 
The language $L(G)$ collects all terminal strings $w$ that show up in final configurations $(w,\mathfrak{f})$ that can be reached by a finite number of steps from the initial configuration as described.
\end{toappendix}

\begin{theoremrep}\label{thm:sscm-0-3--}
$\mathrm{SSCM}(0,*;3;*,*)=\mathrm{RE}$.    
\end{theoremrep}

\begin{proof}
Consider an arbitrary recursively enumerable language $L\subseteq T^*$.
By \cite[Thm. 4]{Feretal07}, there is a graph-controlled grammar $G_C=(\{A,B\},T,(R,\{\mathfrak{i}\},\{\mathfrak{f}\}),A)$ with only two nonterminals $A,B$ and $\mathfrak{i}\neq \mathfrak{f}$ that generates~$L$. Let us assume that the control graph of $G_C$ has~$v$ vertices.
Let $C$ be a new nonterminal, which is also the start symbol of the SSCM grammar $G'$ that we now describe. In our simulation, a subword $C^k$ refers to the vertex $k-1$ of $G_C$ as part of the configuration of~$G_C$ that corresponds to the current sentential form of~$G'$ in the simulation.
We assume $\mathfrak{i}=1$ for the initial vertex of~$G_C$ and $\mathfrak{f}=v$ for the final vertex. As $A$ is the start symbol of~$G_C$, this corresponds to the sentential form $CCA$ of $G'$ that is produced by applying the start matrix to~$C$.  We are adding explanations to the different matrices in order to also show the correctness of the simulation.
For the ease of presentation, let $h$ be the morphism that keeps all terminal symbols and the nonterminal symbols $A,B$ but deletes the (new) nonterminal symbol~$C$. This morphism allows to translate sentential forms of~$G'$ into sentential forms of $G_C$. Likewise, if $C^k$ is a substring of a sentential form~$w$ of $G'$, then $\text{state}(w)$ denotes $k-1$, which should be the state (vertex) of $G_C$ in the current processing. In other words, $w\mapsto (\text{state}(w),h(w))$ should map a sentential form of $G'$ into a corresponding configuration of $G_C$. Even more, we will (always) have $w=C^{\text{state}(w)+1}h(w)$ except at the very beginning, where $w=C$ is the initial sentential form of~$G'$, and at the end, when $w=h(w)\in T^*$.
\begin{itemize}
    \item As a start matrix, we take $$m_{\text{init}}\coloneqq [(C\to CCA),\emptyset,CC]\,.$$
    Notice that $(\text{state}(CCA),h(CCA))=(1,A)$ is the initial configuration of~$G_C$ as  $\mathfrak{i}=1$. Moreover, $CCA=C^{\text{state}(CCA)+1}h(CCA)$ to show that this invariant holds initially. Also, the forbidden context $CC$ will disallow any re-application of $m_{\text{init}}$, because $CC$ will be a substring of any sentential form that is derivable from $CCA$ in $G'$ unless this sentential form is a terminal string. We will maintain this invariant in the following explanations. 
    \item Assume $(\ell:Y_\ell\to\alpha_\ell,\sigma(\ell),\varphi(\ell))$ is a rule of $G_C$ that we want to simulate,  where $Y_\ell\in \{A,B\}$ and $1\leq \ell\leq v$. % and $\{\bar Y\}=\{A,B\}\setminus \{Y\}$. 
    This rule is simulated in two ways in $G'$, differentiating mainly between the success and the failure case in the course of a derivation of $G_C$.
    \begin{itemize}
        \item In the success case, consider $s\in \sigma(\ell)$ and the matrix
        $$m_{\ell\to s}^\sigma\coloneqq [(\langle C\to\emptyword,\rangle^\ell, C\to C^{s+1},Y_\ell\to \alpha_\ell),\emptyset,C^{\ell+2}]\,.$$
        This notation refers to a matrix of length $\ell+2$; the first $\ell$ rules are deletion rules $C\to\emptyword$. If we can successfully apply the first $\ell+1$ rules, we have also checked that there have been at least $\ell+1$ occurrences of~$C$ in the current sentential form. By the forbidden string $C^{\ell+2}$, we also know that there have been at most  $\ell+1$ occurrences of~$C$. Hence, we can only apply the matrix successfully to~$w$ if there are exactly $\ell+1$ occurrences of~$C$ within~$w$ and if there is at least one occurrence of  $Y_\ell\in \{A,B\}$ in~$w$. This reflects that, before applying $m_{\ell\to s}^\sigma$, the simulated grammar was in the state (vertex) $\ell$ and can apply successfully the rule $Y_\ell\to\alpha_\ell$ on the sentential form $h(w)$. Observe that $(\ell,h(w))\Ra_{G_C}(s,u)$ if and only if $w\Ra_{m_{\ell\to s }^\sigma}w'$ with $(\text{state}(w'),h(w'))=(s,u)$. Moreover, if $w=C^{\text{state}(w)+1}h(w)=C^{\ell+1}h(w)$ as an induction hypothesis, $w'=C^{\text{state}(w')+1}h(w')=C^{s+1}u$ which shows that this invariant holds in the induction step for this case. Also, $w'$ contains at least two occurrences of~$C$.
        \item In the failure case, consider  $f\in \varphi(\ell)$ (this will become important only in the next item) and the matrix
        $$m_{\ell}^{\varphi,1}\coloneqq[(C\to\emptyword,C\to C^{g(\ell)+1}),\emptyset,Y_\ell]\,.$$
        This is an intermediate step and moves the number of occurrences of~$C$ out of the usual range by adding $g(\ell)\coloneqq v\ell$ additional $C$'s. The only purpose of this is to check that $Y_\ell$ is not present (hence verifying the failure case in $G_C$) and, by having now more than $v+2$ occurrences of~$C$ disables further applications of any $m_{\ell'\to s'}^\sigma$. We block applying $m_{\ell}^{\varphi,1}$ as the very first matrix by checking that there are indeed at least two occurrences of~$C$ in the current sentential form; this is why we cannot replace the two rules of the matrix by the single rule $C\to C^{g(\ell)}$. In total, if $w$ does not contain $Y_\ell$, $w=C^{k+1}u\Ra_{m_{\ell}^{\varphi,1}} C^{(k+1)+g(\ell)}u$, where $u\in \{A\}^*$ if $Y_\ell=B$ and $u\in \{B\}^*$ if $Y_\ell=A$. Of course, we want to see $k=\ell$, but this is checked only in the next step. Notice that we may apply $m_{\ell}^{\varphi,1}$ multiple times, but no matrix will be able to cope with that many occurrences of $C$'s.
        \item To complete the simulation in the failure case, 
         $$m_{\ell\to f}^{\varphi,2}\coloneqq [(\langle C\to\emptyword,\rangle^{\ell+g(\ell)}, C\to C^{f+1}),\emptyset,C^{\ell+g(\ell)+2}]\,.$$
         This is a matrix of length $\ell+g(\ell)$; this is also the minimum number of occurrences of~$C$'s that we have to find in the current sentential form~$w$. By the forbidden context, we find the subword $C^{\ell+g(\ell)+1}$ in~$w$. This subword (which is actually a prefix of the sentential form due to our inductively proven invariant) is changed into $C^{f+1}$ for $f\in \varphi(\ell)$. Now consider again $w=C^{k+1}u\Ra_{m_{\ell}^{\varphi,1}} C^{(k+1)+g(\ell)}u\Ra_{m_{\ell'}^{\varphi,2}} C^{f+1}u$ for some $f\in \varphi(\ell')$ that enforces $k=\ell=\ell'$ and hence simulates $(\ell,u)\Ra_{G_C}(m,u)$ in $G_C$, also maintaining the other invariants. We already explained our choice of $g(\ell)$ to avoid repeated applications of $m_{\ell}^{\varphi,1}$. Also, observe that if we apply $m_{\ell}^{\varphi,1}$ to some sentential form $C^{k+1}w$  with $k\neq \ell$, then for no $\ell'$, $k+g(\ell)=\ell'+g(\ell')$ is possible as  $k+g(\ell)=k+\ell v $ can equal $\ell'+\ell'v $ only if $k+\ell v= \ell'(v+1)$ and $1\leq k,\ell,\ell'\leq v$. As $k,\ell\leq v$, $k+\ell v$ cannot be a multiple of $v+1$, unless $k=\ell=\ell'$.
         \item Recall that $\mathfrak{f}=v$ is the final vertex. When having reached a final configuration $(\mathfrak{f},w)$, the graph-controlled grammar will collect $w$ into its language if $w\in T^*$. This is the final step that we have to simulate in $G'$. %Observe that so far, the largest prefix $C^k$ of any sentential form that we can deal with later on was $k=2v+1$. Also, there was no way to produce $C^{}$
         This is done with the following matrix.
         $$m_{\text{final}}\coloneqq [(\langle C\to\emptyword,\rangle^{\mathfrak{f}}C\to\emptyword),\emptyset,C^{\mathfrak{f}+2}]$$
         This matrix of length $v+1$ simply deletes the $C$-prefix for the  configuration $(\mathfrak{f},w)$. When succussful, i.e., after $w'=C^{\mathfrak{f}+1}w\Ra_{m_{\text{final}}}w$, the string~$w$ can only contain terminal symbols or $A$ or $B$. But in $G'$, all matrices assume the presence of at least one~$C$. Therefore, such a derivation step can only contribute to the generated language if $w\in T^*$. This shows the correctness  of this final matrix.
    \end{itemize}
\end{itemize}
Finally, notice that if the control graph of the originally given graph-controlled grammar $G_C$ has $v$ many vertices, then the SSCM grammar~$G'$ has a polynomial number of matrices of length polynomial in~$v$. This explains why we could not limit the corresponding descriptional complexity parameters.   
\qed
\end{proof}

\section{Conclusions}

In this paper, we have tried to delineate the Pareto frontier of descriptional complexity for SCM grammars concerning a number of descriptional complexity parameters.
The most natural question here is to complement these results by lower bounds.
Or, can we further lower the upper bounds derived in this paper?
Should we (also) consider alternative parameters, like the number of matrices that actually match the maximum matrix length? We have discussed this through the paper a bit already, but are far from a systematic study.

It is also interesting to observe that in none of the simulations that we propose (but the last one), the sequence in which the matrix rules are applied within a matrix matters. In other words, literally the same descriptional complexity results hold for SCUM grammars, where UM (not formally introduced) should be read as unordered matrix, a model introduced by Cremers and Mayer in~\cite{CreMay73}.
\begin{toappendix}
Yet, generally speaking, unordered matrix grammars look a bit `weaker' than matrix grammars, giving up some control over the derivations that we could not exploit so far. This also may be an indication that some of our results are still not optimal with respect to descriptional complexity. Yet, all this needs more studies and insights. In particular, it is unclear to us if a result similar to Theorem~\ref{thm:sscm-0-3--} will also hold for SCUM grammars.
\end{toappendix}

\bibliographystyle{plain}
\bibliography{ab,hen}

\end{document}